\documentclass[12pt,a4paper]{article}
\usepackage{amsmath,amsfonts,amsthm,amssymb,mathrsfs}
\usepackage{hyperref}
\usepackage{times}
\usepackage{authblk}
\theoremstyle{plain}
\newtheorem{thm}{Theorem}[section]
\newtheorem{lem}[thm]{Lemma}
\newtheorem{cor}[thm]{Corollary}
\newtheorem{prop}[thm]{Proposition}

\theoremstyle{definition}
\newtheorem{defn}{Definition}[section]

\theoremstyle{remark}
\newtheorem*{remark}{Remark}

\numberwithin{equation}{section}

\newcommand{\ga}{\gamma}
\newcommand{\Ga}{\Gamma}
\newcommand{\da}{\delta}
\newcommand{\Da}{\Delta}
\newcommand{\ta}{\theta}

\newcommand{\la}{\lambda}

\newcommand{\fc}{\frac}
\newcommand{\na}{\nabla}

\newcommand{\lt}{\left}
\newcommand{\rt}{\right}

\newcommand{\mf}{\mathbf}
\newcommand{\mb}{\mathbb}
\newcommand{\ml}{\mathcal}

\newcommand{\bs}{\boldsymbol}
\newcommand{\wt}{\widetilde}
\begin{document}
\fontsize{13}{16pt plus.4pt minus.3pt}\selectfont

\title{Proofs of some simplified characterizations of the\\ ground states of spin-1 Bose-Einstein condensates\footnote{This work was partially supported by the National Science Council of the Republic of China under Contract Nos. 99-2115-M-002-003-MY3.}
}
\author[,1]{Liren Lin\footnote{b90201033@ntu.edu.tw}}
\author[,1,2]{
I-Liang Chern\footnote{chern@math.ntu.edu.tw}}
\affil[1]{Department of Mathematics, National Taiwan University, Taipei 106, Taiwan}
\affil[2]{National Center for Theoretical Sciences, Taipei Office, Taiwan}
\date{}
\maketitle
\begin{abstract}
We justify some characterizations of the ground states of spin-1 Bose-Einstein condensates exhibited from numerical simulations. For ferromagnetic systems, we show the validity of the single-mode approximation (SMA). For an antiferromagnetic system with nonzero magnetization, we prove the vanishing of the $m_F=0$ component. In the end of the paper some remaining degenerate situations are also discussed. The proofs of the main results are all based on a simple observation, that a redistribution of masses among different components will reduce the kinetic energy.
\end{abstract}


\section{Introduction}
At ultra low temperature, massive bosons could occupy the same lowest-energy 
state and form the so-called Bose-Einstein condensates (BECs). This phenomenon 
was predicted by Bose and Einstein in 1925, and was first realized on several 
alkali atomic gases in 1995 by laser cooling technique
\cite{Anderson1995.7,Bradley1995.8,Davis1995.11}. 
In early experiments, the atoms were confined in magnetic traps. In this situation the spin degrees of freedom are frozen. Through the mean-field approximation the system is then described by a scalar wave function, which satisfies the Gross-Pitaevskii (GP) equation 
\cite{Dalfovo1999.4,Gross1961,Pitaevskii1961}. In contrast, in an optically trapped atomic BEC all hyperfine spin states can be active simultaneously, and a spin-$F$ BEC is then described by a vector wave function $\Psi=(\psi_{F},\psi_{F-1},\cdots,\psi_{-F})^T$, where the $j$-th component corresponds to the $m_F = j$ hyperfine state \cite{Stamper1998.3,Stenger1998.11,Miesner1999,Barrett01,Gorlitz2003}. The theory of such spinor BEC was first developed independently by several groups \cite{Ohmi1998,Ho1998,Law1998.12}. After these early studies, spinor BEC has become an area of great research interest.

\subsection{Mathematical model for spin-1 BEC}

For a spin-$1$ BEC, the vector wave function 
$\Psi=(\psi_1,\psi_0,\psi_{-1})^T$ satisfies a generalized GP equation:
\begin{equation} \label{GGP}
i\hbar \partial_t \Psi = \fc{\da{E}}{\da\Psi^{*}},
\end{equation}
where the Hamiltonian is given by
\[
E[\Psi] := \int_D\bigg\{ 
\fc{\hbar^2}{2m_a}\sum_j|\na\psi_j|^2
+ V(x)|\Psi|^2 + \fc{c_n}{2}|\Psi|^4 + 
\fc{c_s}{2}\lt|\Psi^{*}S \Psi\rt|^2 
\bigg\}dx.
\]
Here $D$ is a domain in $\mb{R}^d$, $\hbar$ is the reduced Planck constant, $m_a$ is the atomic mass, 
$V$ is a locally bounded real-valued function representing the trap potential, $\Psi^{*}$ is the Hermitian of $\Psi$, and 
$S=(S_x,S_y,S_z)$ is the triple of spin-1 Pauli matrices:
\[
S_x = \fc{1}{\sqrt{2}}
\lt( \begin{array}{ccc}
0 & 1 & 0 \\
1 & 0 & 1 \\
0 & 1 & 0
\end{array}\rt),  \
S_y = \fc{i}{\sqrt{2}}
\lt( \begin{array}{ccc}
0 & -1 & 0 \\
1 & 0 & -1 \\
0 & 1 & 0
\end{array}\rt),  \
S_z =
\lt( \begin{array}{ccc}
1 & 0 & 0 \\
0 & 0 & 0 \\
0 & 0 & -1
\end{array}\rt).
\]
So $\Psi^{*}S\Psi$ denotes the vector
$(\Psi^{*}S_x\Psi,\Psi^{*}S_y\Psi,\Psi^{*}S_z\Psi)$. Also note that 
$|\Psi|$ denotes the Euclidean length $(\sum_j|\psi_j|^2)^{1/2}$, and similiarly for $|\na\psi_j|$ and $|\Psi^{*}S\Psi|$.
The parameters $c_n$ and $c_s$ are real constants given by
\[
c_n = \fc{4\pi \hbar^2}{3m_a}(a_0 + 2a_2), \
c_s = \fc{4\pi \hbar^2}{3m_a}(-a_0 + a_2),
\]
where $a_0$ and $a_2$ are respectively the $s$-wave scattering lengths for scattering channels of total hyperfine spin zero and spin two. The parameter $c_n$ characterizes the spin-independent interaction, and the parameter $c_s$ characterizes the spin-exchange interaction. For $c_n < 0$ (resp. $c_n > 0$), the spin-independent interaction is attractive (resp. repulsive). For $c_s < 0$ (resp. $c_s > 0$), the spin-exchange interaction is ferromagnetic (resp. antiferromagnetic). Typical examples of ferromagnetic and antiferromagnetic systems are $^{87}$Rb and $^{23}$Na condensates.

The generalized GP equation (\ref{GGP}) implies two conserved quantities:
\begin{align*}
\textup{(C1)}\qquad&\int_{D}|\Psi|^2 = N,\\
\textup{(C2)}\qquad&\int_{D}\lt(|\psi_1|^2 - |\psi_{-1}|^2\rt) = M,
\end{align*}
where $N$ is the total number of atoms and $M$ is the total magnetization. For the system to be nontrivial, we assume $N > 0$. We also assume $|M|<N$ 
(note that obviously $|M|\le N$), for if $|M|=N$ the system reduces to a single component BEC, which is a trivial case for all considerations in this work. Now we say $\Psi$ is a ground state if it is a minimizer of $E$ under the above two constraints.

\subsection{Innovation and organization}

In researches concerning ground states of spin-1 BEC, the following ansatz was often adopted:
\[
\psi_j=c_j\psi \quad \textup{for each} \quad j,
\] 
where $c_j$ are constants and $\psi$ is a function independent of $j$. This is called the single-mode approximation (SMA) in the physics literature \cite{Law1998.12,Goldstein1999.5,Pu1999.8,Ho2000.5.1,Pu2000.5.11,Duan2002.2}. It has been found \cite{Yi2002.7} from numerical simulations that ground states obey the SMA exactly for ferromagnetic systems (and does not in general for antiferromagnetic ones), and hence can effectively be characterized as one-component systems. The first goal of this paper is to analytically confirm this observation. On the other hand, for antiferromagnetic systems, we will show that $\psi_0 \equiv 0$ when $M \ne 0$, another well-known phenomenon from numerical simulations \cite{Bao2008.5,CCW2011} not being rigorously proved before. For the degenerate case $M = 0$, however, the SMA is again valid while ground states are not unique, and $\psi_0$ does not necessarily vanish. It's interesting that although the two phenomena (SMA and vanishing of $\psi_0$) look quite irrelevant to each other, they can be proved by the same simple principle, that a redistribution of masses between different components will decrease the kinetic energy. 

The paper is organized as follows. Section 2 is the preliminary, where we reformulate the mathematical model more precisely, and then provide a result of maximum principle which is crucial in justifying the expected characterizations. In Section 2.2 the idea of mass redistribution is introduced. Sections 3 and 4 treat respectively the ferromagnetic and antiferromagnetic systems.


\section{Preliminary}

For notational simplicity, let's redefine
\[
E[\Psi]=\int_D\bigg\{
\sum_j|\na\psi_j|^2 + V|\Psi|^2 + c_n|\Psi|^4 + c_s|\Psi^{*}S\Psi|^2
\bigg\}.
\]
This causes no loss of generality for the phenomena we are going to investigate.
The admissible class is
\begin{align*}
\ml{C} = \lt\{ \Psi\in\big(H^1(D) \cap L^4(D) \cap L^2(D,Vdx)\big)^3 \ \Big| \ 
\Psi \mbox{ satisfies (C1) and (C2)} \rt\},
\end{align*}
where $L^2(D,Vdx)$ consists of all functions $f$ such that 
$\int_D V|f|^2<\infty$. Let $\mf{u}$ denotes $(u_1,u_0,u_{-1})$. We also define
\begin{align*}
\ml{A} &= \lt\{ \mf{u} \in \ml{C} \ | \ 
u_j \ge 0 \mbox{ for each } j \rt\}; \\
\ml{A}_1 &= \lt\{ \mf{u} \in \ml{A} \ | \ 
\mf{u}=(\ga_1 f,\ga_0 f,\ga_{-1} f)
\mbox{ for some constants }\ga_j\mbox{ and some function }f\rt\}; \\
\ml{A}_2 &= \lt\{ \mf{u} \in \ml{A} \ | \ 
u_0 \equiv 0 \rt\}.
\end{align*}
Let's also use $\bs{\ga}$ to denote $(\ga_1,\ga_0,\ga_{-1})$, so $(\ga_1 f,\ga_0 f,\ga_{-1} f)$ can be abbreviated as $\bs{\ga}f$. 

In Section 2.1, we introduce a common reduction which shows that to study ground states we can simply consider $\ml{A}$ instead of $\ml{C}$. Indeed, $\ml{A}$ consists just the amplitudes of elements in $\ml{C}$. And $\ml{A}_1$ (resp. $\ml{A}_2$) corresponds to the set of all elements obeying the SMA (resp. with vanishing zeroth components). For the moment, we do not consider any boundary condition for simplicity. See the remark in the end of Section 3. 

\subsection{Reduction from $\ml{C}$ to $\ml{A}$}

Given $\Psi\in\ml{C}$. Let $u_j e^{i\ta_j}$ be the polar form of $\psi_j$. It's easy to check that if $\Psi$ is a ground state, that is $\Psi$ minimizes $E$ over $\ml{C}$, then the $\ta_j$'s are constants satisfying
\begin{equation}\label{theta}
\cos{(\ta_1-2\ta_0+\ta_{-1})} = \pm 1 \quad\textup{for}\quad c_s \lessgtr 0,
\end{equation}
and 
\begin{equation}\label{H}
E[\Psi] = \int_D\bigg\{\sum_j|\na u_j|^2 + V|\mf{u}|^2 + c_n|\mf{u}|^4 
+ c_s \Big[ 2u_0^2(u_1\pm u_{-1})^2 + (u_1^2-u_{-1}^2)^2\Big]\bigg\},
\end{equation}
where the plus-minus sign $\pm$ corresponds to $c_s\lessgtr 0$. Let's now define 
$\mb{E}:\ml{A}\to\mb{R}$, $\mb{E}[\mf{u}]$ is given by the right-hand side of (\ref{H}). What we claimed is if $\Psi$ is a ground state, then
$E[\Psi]=\mb{E}[(|\psi_1|,|\psi_0|,|\psi_{-1}|)]$. Conversely, if any $\mf{u}\in\ml{A}$ satisfies
\[
\mb{E}[\mf{u}]=\min_{\mf{v}\in\ml{A}}\mb{E}[\mf{v}],
\]
the vector $\Psi$ defined by $\psi_j=u_j e^{i\ta_j}$ is a ground state as long as the $\ta_j$'s are constants satisfying (\ref{theta}). Thus, studying ground states of $E$ is equivalent to studying minimizers of $\mb{E}$. Without loss of generality, we will henceforth consider $\mb{E}$ instead of the original $E$. 

For convenience let's use $H$ to denote the integrand of $\mb{E}$, i.e. 
$\mb{E}[\mf{u}]=\int_D H(\mf{u})$. We also write
$H = H_1 + H_2$, where
\begin{align*}
H_1(\mf{u}) &= \sum_j|\na u_j|^2 + c_s\Big[ 2u_0^2(u_1 \pm u_{-1})^2 + (u_1^2 - u_{-1}^2)^2 \Big], \\
H_2(\mf{u}) &= V|\mf{u}|^2 + c_n|\mf{u}|^4.
\end{align*}
This splitting of $H$ has no physical meaning but only for convenience of later discussion. 

We shall denote the set of all minimizers of $\mb{E}$ over $\ml{A}$ by $\ml{G}$.
The Euler-Lagrange system for $\mf{u}\in\ml{G}$ is given by the following coupled Gross-Pitaevskii equations:
\begin{equation}\label{ELeq}
\lt\{
\begin{aligned}
(\mu + \lambda) u_1    &= \ml{L}u_1 + 2c_s\lt[u_0^2(u_1\pm u_{-1})+u_1(u_1^2-u_{-1}^2)\rt]\\
\mu u_0                &= \ml{L}u_0 + 2c_s u_0 (u_1\pm u_{-1})^2\\
(\mu - \lambda) u_{-1} &= \ml{L}u_{-1} + 2c_s\lt[u_0^2(u_{-1}\pm u_1)+u_{-1}(u_{-1}^2-u_1^2)\rt],
\end{aligned}\rt.
\end{equation}
where 
$\ml{L}=-\Da + V + 2c_n|\mf{u}|^2$, $\la$ and $\mu$ are the Lagrange multipliers. We remark that in this paper we do not involve ourselves in the problem of existence. To best illustrate the simplicity of our method, we just assume there is a ground state. (see \cite{lieb-2000,cao_2011,bao_cai_2011} for related concerns of existence problem). Also note that $\mf{u}\in\ml{G}$ is continuously differentiable by standard regularity theorem (see e.g. \cite[10.2]{LiebLoss}).

The following lemma will be crucial in our characterizations of ground states.

\begin{lem}\label{lemma}
If $\mf{u}\in\ml{G}$, then for each $j$, either $u_j \equiv 0$ or $u_j>0$ 
on all of $D$.
\end{lem}
\begin{proof}
For an arbitrary compact $K \subset D$, by subtracting respectively $Q_ju_j$, $j=1,0,-1$, from the three equations in (\ref{ELeq}) with large enough constants $Q_j$, and using the assumption $u_j \ge 0$, it's easy to verify that each $u_j$ satisfies
\[
\Da u_j + h_j u_j \le 0
\]
for some $h_j\le 0$ on $K$. Thus either $u_j>0$ or $u_j \equiv 0$ on $K$ by strong maximum principle. Since $K \subset D$ is arbitrary, the assertion of the lemma holds.
\end{proof}

\subsection{A kinetic-energy-reduced redistribution}\label{redis}

Consider an $n$-tuple of nonnegative functions $\mf{f}=(f_1,f_2,...,f_n) \in (H^1(D))^n$. Let $g=|\mf{f}|$.
It's well-known that  $|\na g|^2\le\sum_k|\na f_k|^2$. In fact, direct computation gives
\begin{equation}\label{equal}
\sum_k|\na f_k|^2 - |\na g|^2 = 
\lt\{\begin{aligned}
&\fc{1}{g^2}\sum_{j < k}|f_j\na f_k - f_k\na f_j|^2
&\mbox{ on where }g>0&\\
&0
&\mbox{ on where }g=0&.
\end{aligned}\rt.
\end{equation}
In particular, $|\na g|^2 = \sum_k|\na f_k|^2$ if and only if 
$f_j \na f_k - f_k \na f_j = 0$ for $j \ne k$.

The property above has a simple while interesting generalization, when $f_k^2$ ($k=1,2,...,n$) do not sum to a single $g^2$, but instead are redistributed into multiple parts. To be precise, we give the following definition.

\begin{defn}\label{MR}
Let $\mf{f}$ be as above, and
let $\mf{g}=(g_1,g_2,...,g_m)$ be an $m$-tuple of nonnegative functions. 
We say $\mf{g}$ is a mass redistribution of $\mf{f}$ if 
\begin{align*}
g_{\ell}^2=\sum_{k=1}^{n} a_{\ell k}f_k^2
\end{align*}
for each $\ell=1,...,m$, where $a_{\ell k}$ are nonnegative constants satisfying $\sum_{\ell=1}^m a_{\ell k}=1$, for each $k=1,...,n$.
\end{defn}

Note that $g=|\mf{f}|$ is the only mass redistribution of 
$\mf{f}$ with $m=1$. In general we have the following proposition.

\begin{prop}\label{prop2.2}
For any mass redistribution $\mf{g}$ of $\mf{f}$ as in \textup{Definition \ref{MR}}, we have 
\begin{itemize}
\item[\textup{(1)}]
$|\mf{g}|=|\mf{f}|$\textup{;}
\item[\textup{(2)}]
$\sum_{\ell=1}^m|\na{g}_{\ell}|^2 \le \sum_{k=1}^n|\na{f_k}|^2$. Equality holds if and only if $f_j\na f_k - f_k\na f_j = 0$ for each $j \ne k$ with $a_{\ell j}a_{\ell k}\ne 0$ for at least one $\ell$.
\end{itemize}
\end{prop}

\begin{proof}
The first assertion follows directly from the definition of redistribution.
For (2), from (\ref{equal}) we have
\[
\sum_k a_{\ell k}|\na f_k|^2 - |\na g_{\ell}|^2 =
\lt\{\begin{aligned}
&\fc{1}{g_{\ell}^2}\sum_{j < k}a_{\ell j}a_{\ell k}|f_j\na f_k - f_k\na f_j|^2
&\textup{ on where }g_{\ell}>0&\\
&0
&\textup{ on where }g_{\ell}=0&,
\end{aligned}\rt.
\]
and the assertion is obtained by summing over all $\ell$.
\end{proof}

In this work, we will consider mass redistributions of $\mf{u}\in\ml{A}$.
Since the square of the amplitude of a wave function represents its mass distribution, we use the adjective ``mass'' to stress that it's a redistribution of the squares\footnote{And one can naturally generalize the idea to $p$-th power redistribution, which may be useful in studying systems with $p$-Laplacian terms.}. To save notation, in the following we shall omit it and simply say ``redistribution''. 
Note that if $\mf{u}\in\ml{A}$ and $\mf{v}=(v_1,v_0,v_{-1})$ is a redistribution of $\mf{u}$, then from (1) of Proposition \ref{prop2.2},
$\mf{v}$ satisfies (C1) automatically and $H_2(\mf{v}) \equiv H_2(\mf{u})$. These facts together with (2) of the proposition allow us to give unified and simple justifications of the two properties mentioned in the introduction. 


\section{Ferromagnetic systems}

In this section we assume $c_s<0$, and the goal is to prove the validity of SMA.
That is we want to show $\ml{G}\subset\ml{A}_1$. The idea is to find, for 
$\mf{u}\in\ml{A}$, a redistribution of $\mf{u}$ in $\ml{A}_1$ that has no 
larger energy than $\mf{u}$, and then try to conclude that $\mf{u}$ must itself be the redistributed element provided $\mf{u}\in\ml{G}$. 

Now given any $\mf{u}\in\ml{A}$. It's easy to see that a redistribution of $\mf{u}$ in $\ml{A}_1$ can be expressed as $\bs{\ga}|\mf{u}|$, where 
$\bs{\ga}=(\ga_1,\ga_0,\ga_{-1})$ is any triple of nonnegative constants satisfying
\begin{equation} \label{adconst}
\lt\{\begin{aligned}
&\ga_1^2+\ga_0^2+\ga_{-1}^2 = 1\\
&\ga_1^2 - \ga_{-1}^2 = M/N.
\end{aligned}\rt.
\end{equation}
Let $\Ga$ denote the set containing all such $\bs{\ga}$: 
\[
\Ga := \lt\{ \bs{\ga}\in\mb{R}^3 \ \big| \ \ga_j \ge 0 \mbox{ for each }j,\  \bs{\ga} \mbox{ satisfies } (\ref{adconst}) \rt\}.
\]
Then 
$H_2(\bs{\ga}|\mf{u}|) \equiv H_2(\mf{u})$ for each $\bs{\ga}\in\Ga$. On the other hand, 
\[
H_1(\bs{\ga}|\mf{u}|) = |\na|\mf{u}||^2 + c_s P(\bs{\ga})|\mf{u}|^4,
\]
where
\begin{align*}
P(\bs{\ga})=2\ga_0^2(\ga_1 + \ga_{-1})^2+\fc{M^2}{N^2}.
\end{align*}
Since $c_s<0$, 
\[
\min_{\bs{\ga}\in\Ga} \lt( c_s P(\bs{\ga})|\mf{u}|^4 \rt)
= c_s \bigg(\max_{\bs{\ga}\in\Ga}P(\bs{\ga})\bigg)|\mf{u}|^4,
\] 
and it's easy to check that 
\[
\max_{\bs{\ga}\in\Ga}P(\bs{\ga})=P(\bs{\ga}^{\star})=1,
\]
where the maximizer $\bs{\ga}^{\star}=(\ga_1^{\star},\ga_0^{\star},\ga_{-1}^{\star})$ is unique and is given by
\begin{align*}
\ga_1^{\star} = \fc{1}{2}\lt(1+\fc{M}{N}\rt),\
\ga_0^{\star} = \sqrt{ \fc{1}{2} \lt( 1-\fc{M^2}{N^2} \rt) },\mbox{ and }
\ga_{-1}^{\star} = \fc{1}{2}\lt( 1-\fc{M}{N} \rt).
\end{align*}
On the other hand, we already know $|\na|\mf{u}||^2\le\sum_j|\na u_j|^2$. Thus we have proved
\begin{align}\label{lem1}
H(\bs{\ga}^{\star}|\mf{u}|) = |\na|\mf{u}||^2 + c_s|\mf{u}|^4 \le
H(\mf{u})
\end{align}
for any $\mf{u} \in \ml{A}$. We can now prove the first characterization of ground states by examining the condition for equality of (\ref{lem1}).

\begin{thm}\label{thm1}
Assume $c_s<0$. If $\mf{u}\in\ml{G}$, then $\mf{u} = \bs{\ga}^\star|\mf{u}|$.
\end{thm}
\begin{proof}
\begin{align*}
H(\mf{u})-H(\bs{\ga}^{\star}|\mf{u}|)
&=H_1(\mf{u})-H_1(\bs{\ga}^{\star}|\mf{u}|)\\
&=\big(\sum_j|\na u_j|^2 - |\na|\mf{u}||^2\big) - c_s \lt\{|\mf{u}|^4-\Big[2u_0^2(u_1 + u_{-1})^2 + (u_1^2-u_{-1}^2)^2\Big]\rt\}\\
&=\big(\sum_j|\na u_j|^2 - |\na|\mf{u}||^2\big) - c_s(u_0^2-2u_1u_{-1})^2
\end{align*}
by direct calculation. If $\mf{u}\in\ml{G}$, we must have 
$\mb{E}[\mf{u}]=\mb{E}[\bs{\ga}^{\star}|\mf{u}|]$, which implies 
$H(\mf{u})=H(\bs{\ga}^{\star}|\mf{u}|)$, and hence
\begin{align}
u_j\na u_k - u_k\na u_j &= 0\quad\textup{for}\quad j\neq k\,;\label{15}\\
u_0^2 - 2u_1u_{-1} &= 0.\label{16}
\end{align}
Since we assume the total number of atoms $N>0$, from Lemma \ref{lemma} at least one $u_j$ is strictly positive in $D$. Assume $u_1>0$ on $D$. Then from (\ref{15}) we have 
\begin{equation}\label{ezero}
\na\lt(\fc{u_0}{u_1}\rt) = \na\lt(\fc{u_{-1}}{u_1}\rt)=0.
\end{equation}
Since $D$ is connected, (\ref{ezero}) implies $u_0$ and $u_{-1}$ are both constant multiples of $u_1$. This shows $\mf{u} \in \ml{A}_1$, and (ii) follows either by (\ref{16}) or by the fact that $\bs{\ga}^{\star}$ is the unique maximizer of $P$ over $\Ga$. The case $u_0>0$ and $u_{-1}>0$ can be proved similarly.
\end{proof}

The above theorem implies that searching for ground states of ferromagnetic spin-1 BEC can be reduced to a ``one-component'' minimization problem. 
Precisely, let 
\begin{equation}\label{sing}
\ml{A}^s 
=\lt\{|\mf{u}|\ | \ \mf{u}\in\ml{A}\rt\}
=\lt\{ f \in H^1(D) \cap L^4(D)\cap L^2(D,Vdx) \ \big| \ 
\textstyle{\int_D f^2 = N} \rt\},
\end{equation}
and define $\mb{E}^s:\ml{A}^s\to\mb{R}$,
\begin{align*}
\mb{E}^s[f] = \int_D \lt\{|\na f|^2 + Vf^2 + (c_n + c_s)f^4\rt\}.
\end{align*}
Then $\mb{E}[\bs{\ga}^{\star}f]=\mb{E}^s[f]$ for $f\in\ml{A}^s$. 
Also let
\begin{align*}
\ml{G}^s = \lt\{
f\in\ml{A}^s \ \lt| \ \mb{E}^s[f] = \min_{g\in\ml{A}^s}\mb{E}^s[g]
\rt.\rt\}.
\end{align*}
Then if $\mf{u}\in\ml{G}$, $\mf{u}=\bs{\ga}^{\star}|\mf{u}|$, and hence
\begin{align*}
\mb{E}^s[|\mf{u}|]=\mb{E}[\bs{\ga}^{\star}|\mf{u}|]  
\le \mb{E}[\bs{\ga}^{\star}f] =\mb{E}^s[f]
\end{align*}
for every $f\in\ml{A}^s$. Thus $|\mf{u}|\in\ml{G}^s$. Conversely if 
$f\in\ml{G}^s$, then 
\begin{align*}
\mb{E}[\bs{\ga}^{\star}f]=\mb{E}^s[f]
\le\mb{E}^s[|\mf{u}|]=\mb{E}[\bs{\ga}^{\star}|\mf{u}|]
\le\mb{E}[\mf{u}]
\end{align*}
for every $\mf{u}\in\ml{A}$, and hence $\bs{\ga}^{\star}f\in\ml{G}$. We thus obtain the following characterization of $\ml{G}$.

\begin{cor}\label{cor3.2}
$\ml{G} = \lt\{ \bs{\ga}^{\star}f \ | \ f\in\ml{G}^s\rt\}$.
\end{cor}

\begin{remark}
We can add more assumptions in the definition of $\ml{A}$. The only thing we need to take care is that we need $\bs{\ga}^{\star}|\mf{u}|\in\ml{A}$ whenever $\mf{u}\in\ml{G}$, so that $\mb{E}[\mf{u}]\le\mb{E}[\bs{\ga}^{\star}|\mf{u}|]$ is not violated. In particular, in case that a homogeneous boundary condition (e.g. homogeneous Dirichlet or Neumann boundary condition) is considered, the induced boundary condition for $\bs{\ga}^{\star}|\mf{u}|$ is also homogeneous of the same kind, and Theorem \ref{thm1} (and hence Corollary \ref{cor3.2}) remains valid.
\end{remark}


\section{Antiferromagnetic systems and some degenerate cases}

The main focus of this section is the phenomenon $u_0\equiv 0$. After justifying it in Section 4.1, some degenerate situations are also discussed in Section 4.2.

\subsection{Justification of the vanishing phenomenon}

Assume $c_s>0$ in this subsection.
We want to show that any ground state must have a vanishing zeroth component, and hence is a two-component BEC. Similar to the approach in the previous section, we want to find an appropriate redistribution 
$\wt{\mf{u}}\in\ml{A}_2$ of $\mf{u}\in\ml{A}$ so that 
$\mb{E}[\wt{\mf{u}}]\le\mb{E}[\mf{u}]$. Now, not as before, the assumption
$\wt{\mf{u}}\in\ml{A}_2$ doesn't give rise to a definite candidate of 
$\wt{\mf{u}}$. In view that such $\wt{\mf{u}}$ satisfies 
$|\wt{\mf{u}}|=|\mf{u}|$ and hence (C1), as a guess, we try imposing the additional assumption that $\wt{\mf{u}}$ also satisfies 
\[
\wt{u}_{1}^2 - \wt{u}_{-1}^2 = u_1^2 - u_{-1}^2,
\]
so that (C2) is also satisfied by $\wt{\mf{u}}$ automatically. This results in only one possibility, that is 
\begin{equation}\label{utilde}
\wt{u}_j = \sqrt{u_j^2+\fc{u_0^2}{2}}\quad\textup{for}\quad j=1,-1.
\end{equation}
For any $\mf{u} \in \ml{A}$, we then let $\wt{\mf{u}}\in\ml{A}_2$ be its redistribution defined by \textup{(\ref{utilde})}. It's fortunate that it works.
In fact,
\begin{align}\label{lem2}
H(\mf{u})-H(\wt{\mf{u}})=\big(\sum_j|\na u_j|^2-\sum_j|\na\wt{u}_j|^2\big) + 2c_su_0^2(u_1-u_{-1})^2\ge 0,
\end{align}
and we have the following analogue of Theorem \ref{thm1}.

\begin{thm}\label{thm2}
Assume $c_s > 0$ and $M \ne 0$, then
$\mf{u}\in\ml{G}$ implies $\mf{u} = \wt{\mf{u}}$.
\end{thm}

\begin{proof}
Assume $\mf{u}\in\ml{G}$. From (\ref{lem2}), 
$\mb{E}[\mf{u}]=\mb{E}[\wt{\mf{u}}]$, and
\begin{align*}
u_0^2(u_1-u_{-1})^2 \equiv 0.
\end{align*}
By Lemma \ref{lemma}, we have either $u_0 \equiv 0$ or $u_1 \equiv u_{-1}$. However since we assume $M \ne 0$, we cannot have $u_1 \equiv u_{-1}$, and the assertion follows.
\end{proof}

\subsection{Some degenerate situations}

The requirement $M \ne 0$ in Theorem \ref{thm2} is necessary. In fact, for 
$M = 0$, SMA is again valid while ground states are not unique, 
and $u_0 \equiv 0$ is not necessarily the case. Precisely, consider the minimization problem (recall that $\ml{A}^s$ is defined by (\ref{sing}))
\begin{equation}\label{degm}
\min_{f\in\ml{A}_s}\int_D \lt\{ |\na f|^2 + Vf^2 + c_n f^4 \rt\},
\end{equation}
and we have the following characterization.

\begin{thm}
If $c_s>0,M = 0$ or $c_s = 0$, then
\[
\ml{G} = \lt\{\lt.\big(t,\sqrt{1-2t^2},t\big)f \ \rt| \
0\le t\le 1/\sqrt{2},\, f \mbox{ is a } \mbox{ solution of }(\ref{degm})
\rt\}.
\]
\end{thm}

\begin{proof}
Note that since $M = 0$, from (\ref{adconst}),  $\bs{\ga}\in\Ga$ implies 
\[
\bs{\ga} = \Big(t,\sqrt{1-2t^2},t\Big) \quad \mbox{for some}\quad t \in \lt[0,1/\sqrt{2}\,\rt].
\]
Now it's easy to see that for any $\mf{u}\in\ml{A}$ and $\bs{\ga}\in\Ga$ we have
\[
H(\bs{\ga}|\mf{u}|) = |\na|\mf{u}||^2 + V|\mf{u}|^2 + c_n|\mf{u}|^4,
\]
which satisfies $H(\bs{\ga}|\mf{u}|) \le H(\mf{u})$ by Proposition \ref{prop2.2}, and the remaining of the proof is the same as in Section 3. 
\end{proof}

In contrast to the above theorem, the following corollary of Theorem \ref{thm2} shows that SMA is almost never the case when $M\ne 0$.

\begin{cor}
Assume $c_s>0$ and $M\ne 0$, then
$\mf{u} \in \ml{G} \cap \ml{A}_1$ implies 
$u_1$, $u_{-1}$ and $V$ are constants.
\end{cor}

\begin{proof}
By Theorem \ref{thm2}, the Euler-Lagrange system (\ref{ELeq}) is reduced to the following two-component system:
\begin{equation}
\lt\{\begin{aligned}\label{345}
(\mu + \la) u_1 &= \ml{L}u_1 + 2c_s u_1(u_1^2-u_{-1}^2)\\
(\mu - \la) u_{-1} &= \ml{L}u_{-1} + 2c_s u_{-1}(u_{-1}^2-u_1^2),
\end{aligned}\rt.\end{equation}
where $\ml{L}=-\Delta + V + 2c_n(u_1^2 + u_{-1}^2)$.

Recall that we assume $-N<M<N$, thus, for $j=1,-1$, $u_j>0$ on $D$. So $\mf{u} \in \ml{A}_1$ implies $u_{-1}=\kappa u_{1}$ for some constant $\kappa>0$. Also note that $\kappa\ne 1$ since $M\ne 0$. The system (\ref{345}) then gives the following two equations for $u_1$:
\begin{align}
(\mu + \lambda) u_1&=-\Delta u_1 + Vu_1 + 2c_n(1+\kappa^2)u_1^3 + 2c_s(1-\kappa^2)u_1^3;\label{e1}\\
(\mu - \lambda) u_1&=-\Delta u_1 + Vu_1 + 2c_n(1+\kappa^2)u_1^3 + 2c_s(\kappa^2-1)u_1^3.\label{e2}
\end{align}
Now (\ref{e1}) minus (\ref{e2}) gives $\la u_1 = 2c_s(1-\kappa^2)u_1^3$.
Since $u_1>0$ on $D$, we get 
\[
u_1 = \sqrt{\fc{\la}{2c_s(1-\kappa^2)}}\,.
\] 
In particular $u_1$ and $u_{-1}=\kappa u_1$ are constants. Hence 
$\Da u_1= 0$, and then (\ref{e1}) plus (\ref{e2}) gives
\[
\mu u_1 = V u_1 + 2c_n(1+\kappa^2)u_1^3,
\]
from which we get
\[
V=\mu - 2c_n(1+\kappa^2)u_1^2 = \mu- \fc{c_n(1+\kappa^2)}{c_s(1-\kappa^2)}\lambda,
\]
which is also a constant.
\end{proof}

\bibliographystyle{plain}
\bibliography{}

\end{document}